\newif\iflipics
\lipicsfalse

\iflipics
\documentclass[a4paper,USenglish]{lipics-v2016}
\usepackage{microtype}
\usepackage{xspace}
\usepackage[ruled,vlined,linesnumbered]{algorithm2e}

\else
\documentclass[11pt]{article}
\usepackage[margin=1in]{geometry}
\usepackage{xcolor,xspace}
\usepackage{amsfonts,amsmath,amssymb,amsthm}
\usepackage{cite,comment,enumerate,hyperref}
\usepackage[ruled,vlined,linesnumbered]{algorithm2e}
\usepackage{algorithmic}
\newtheorem{theorem}{Theorem}
\newtheorem{corollary}[theorem]{Corollary}

\newtheorem{lemma}[theorem]{Lemma}

\theoremstyle{definition}
\newtheorem{definition}{Definition}

\theoremstyle{remark}

\newcounter{note}[section]

\fi

\newcommand{\LOCAL}{$\mathcal{LOCAL}$\xspace}
\newcommand{\BASIC}[1][k]{\textsc{Basic $#1$-Spanner}\xspace}
\newcommand{\DIRECTED}[1][k]{\textsc{Directed $#1$-Spanner}\xspace}
\newcommand{\DEGREE}[1][k]{\textsc{Lowest-Degree $#1$-Spanner}\xspace}

\newcommand{\DSNDISTANCE}{\textsc{Directed Steiner Network with Distance Constraints}\xspace}
\newcommand{\SLSN}{\textsc{Shallow-Light Steiner Network}\xspace}

\iflipics
\title{Distributed Distance-Bounded Network Design Through Distributed Convex Programming\footnote{Supported in part by NSF awards 1464239 and 1535887.}}
\titlerunning{Distributed Network Design Through Convex Programming}
\author[1]{Michael Dinitz}
\author[2]{Yasamin Nazari}
\affil[1]{Johns Hopkins University, Baltimore, MD, USA\\
  \texttt{mdinitz@cs.jhu.edu}}
\affil[2]{Johns Hopkins University, Baltimore, MD, USA\\
  \texttt{ynazari@jhu.edu}}
\authorrunning{M.~Dinitz and Y.~Nazari} 

\Copyright{Michael Dinitz and Yasamin Nazari}

\subjclass{F.2.2 Nonnumerical Algorithms and Problems}
\keywords{distributed algorithms, approximation algorithms, convex programming}

\EventEditors{John Q. Open and Joan R. Acces}
\EventNoEds{2}
\EventLongTitle{42nd Conference on Very Important Topics (CVIT 2016)}
\EventShortTitle{OPODIS 2017}
\EventAcronym{OPODIS}
\EventYear{2016}
\EventDate{December 24--27, 2016}
\EventLocation{Little Whinging, United Kingdom}
\EventLogo{}
\SeriesVolume{42}
\ArticleNo{23}

\else
\title{Distributed Distance-Bounded Network Design Through Distributed Convex Programming\footnote{Supported in part by NSF awards 1464239 and 1535887.}}
\author{Michael Dinitz \\
Johns Hopkins University\\
mdinitz@cs.jhu.edu
\and
Yasamin Nazari\\
Johns Hopkins University\\
ynazari@jhu.edu
}

\fi

\begin{document}
\maketitle

\begin{abstract}
Solving linear programs is often a challenging task in distributed settings. While there are good algorithms for solving packing and covering linear programs in a distributed manner (Kuhn et al.~2006), this is essentially the only class of linear programs for which such an algorithm is known. In this work we provide a distributed algorithm for solving a different class of convex programs which we call ``distance-bounded network design convex programs''.  These can be thought of as relaxations of network design problems in which the connectivity requirement includes a distance constraint (most notably, graph spanners).  Our algorithm runs in $O( (D/\epsilon) \log n)$ rounds in the $\mathcal{LOCAL}$ model and finds a $(1+\epsilon)$-approximation to the optimal LP solution for any $0 < \epsilon \leq 1$, where $D$ is the largest distance constraint.  

While solving linear programs in a distributed setting is interesting in its own right, this class of convex programs is particularly important because solving them is often a crucial step when designing approximation algorithms.  Hence we almost immediately obtain new and improved distributed approximation algorithms for a variety of network design problems, including Basic $3$- and $4$-Spanner, Directed $k$-Spanner, Lowest Degree $k$-Spanner, and Shallow-Light Steiner Network Design with a spanning demand graph.  Our algorithms do not require any ``heavy'' computation and essentially match the best-known centralized approximation algorithms, while previous approaches which do not use heavy computation give approximations which are worse than the best-known centralized bounds.  
\end{abstract}

\section{Introduction}

Distributed network design is a classical type of distributed algorithmic problem, going back at least to the seminal work on distributed MST by Gallager, Humblet, and Spira~\cite{GHS83}.  By ``network design'', we mean the class of problems which can be phrased as ``given input graph $G$, find a subgraph $H$ which has some property $P$, and minimize the cost of $H$''.  Clearly different properties $P$, and different notions of cost, lead to very different problems.  One important class of problems are \emph{distance-bounded} network design problems, where the property $P$ is that certain pairs of vertices are within some distance of each other in $H$ (where distance refers to the shortest-path distance).  The most well-known type of distance-bounded network design problems are problems involving \emph{graph spanners}, in which the distance requirement is that the distance in $H$ for all (or certain) pairs is within a certain factor (known as the stretch) of their original distance in $H$.  But there are many other important versions of distance-bounded network design, such as the bounded diameter problem~\cite{DK99} and the shallow-light Steiner tree/network problems~\cite{KS16}.  

Many of these problems are NP-hard, so they cannot be solved optimally in polynomial time even in the centralized setting.  Thus they have been studied extensively from an approximation algorithms point of view, where we design algorithms which approximate the optimal solution but which run in polynomial time.  For many of these problems, a key step in the best-known centralized approximation algorithm is solving a linear programming relaxation of the problem, and then rounding the optimal fractional solution into a feasible integral solution.  Interestingly, it is relatively common for the rounding to be ``local'': if we are in a distributed setting and happen to know the optimal fractional LP solution, then the algorithm used to round this to an integral solution can be accomplished with a tiny amount of extra time (either $0$ or a small constant number of rounds).  So the bottleneck when trying to make these algorithms distributed is solving the LP, not rounding it.

Solving LPs in distributed settings has received only a small amount of attention, since it unfortunately turns out to be extremely challenging in general.  Most notably, Kuhn, Moscibroda, and Wattenhofer~\cite{KMW06} gave an efficient distributed algorithm (in the \LOCAL model of distributed computation) for packing/covering LPs.  Unfortunately, the LPs used for distance-bounded network design are not packing/covering LPs\footnote{They can be turned into packing/covering LPs through a projection operation, but unfortunately this technique results in an exponential number of constraints, making~\cite{KMW06} inapplicable.  However, this technique has been used in the centralized setting for the fault-tolerant directed $k$-spanner problem~\cite{DK-STOC}}, and hence we are not able to use their techniques.  In this paper we show how to solve these LPs (and convex generalizations of them) in the \LOCAL model of distributed computation, which almost immediately gives the best-known results for a variety of distance-bounded network design problems.  

In particular, for many network design problems (\DIRECTED, \BASIC[3], \BASIC[4], \DEGREE, \DSNDISTANCE with spanning demands, and \SLSN with spanning demands) we give approximation algorithms which run in $O(D \log n)$ rounds (where $D$ is the maximum distance bound) and have the same approximation ratios as in the centralized setting.  Previous distributed algorithms for these problems with similar round complexity have either used ``heavy'' computations (non-polynomial time algorithms) at the nodes (in which case they can often do \emph{better} than the best computationally-bounded centralized algorithm), or give approximation bounds which are asymptotically worse than the best centralized bounds.  See Section~\ref{sec:related} for more discussion of previous work.

\subsection{Our Results}

We give two main types of results.  First, we give a distributed algorithm that (approximately) solves distance-bounded network design convex programs with small round complexity.  We then use this result to (almost immediately) get improved distributed approximation algorithms for a variety of network design problems.  

\subsubsection{Solving convex programs}
Stating our main technical result (distributed approximations of distance-bounded network design convex programs) in full generality requires significant technical setup, so we provide an informal description here.  See Section~\ref{sec:solving-convex} for the full definitions and theorem statements (Theorem~\ref{thm:distributed_LP} in particular).  But informally, a distance-bounded network design convex program is the following.  We are given a graph $G = (V, E)$, a set $\mathcal S \subseteq V \times V$, and for each $(u,v) \in \mathcal S$ there is a set of ``allowed'' $u-v$ paths $\mathcal P_{u,v}$.  Informally, the integral problem is to find a subgraph $H$ of $G$ so that every $(u,v) \in \mathcal S$ is connected by at least one path from $\mathcal P_{u,v}$ in $H$, and the goal is to minimize some notion of ``cost''.  If our notion of ``cost'' is captured by an objective function $g : \mathbb{R}_{\geq 0}^{|E|} \rightarrow \mathbb{R}$ (which is typically linear, but which can be more general convex functions as long as they satisfy a ``partitionability'' constraint -- see Section~\ref{sec:solving-convex} for the details), then the natural relaxation of this problem is the following convex program, which has a variable $x_e$ for every edge and a variable $f_P$ for every allowed path.

\begin{align*} 
\min \quad &g({x})\\
\text{s.t.} \quad &\sum_{P \in \mathcal{P}_{u,v}: e\in P} f_P\leq x_e  &\forall(u,v) \in  \mathcal{S}, \forall e \in E \\
&\sum_{P \in \mathcal{P}_{u,v}} f_P \geq 1 &\forall(u,v) \in  \mathcal{S}\\
&x_e \geq 0 &\forall	e \in E\\
&f_P \geq 0 &\forall(u,v)\in  \mathcal{S}, \forall P\in \mathcal{P}_{u,v}
\end{align*}

Informally, the first type of constraint says that an allowed path is included only if all edges in it are included, and the second type of constraint required us to include at least one allowed path for each $(u,v) \in \mathcal S$.  We call this type of convex program a \emph{distance-bounded network design convex program}.  It is clearly not a packing/covering LP due to the first type of constraint, and hence there is no known distributed algorithm to solve this kind of program.  However, note that if the maximum length of any allowed path is constant, then there are only a polynomial number of such paths, and hence the size of the convex program is polynomial and so it can be solved in polynomial time in the centralized setting under reasonable assumptions on $g$ (see~\cite{GLS88} for details on solving convex programs in polynomial time). 

Our main technical result is that we can approximately solve these optimization problems even in a distributed setting. For any path $P$ let $\ell(P)$ denote the length of the path (the number of edges in it).  
\begin{theorem} \label{thm:main-informal}
For any constant $\epsilon > 0$, any distance-bounded network design convex program can be solved up to a $(1+\epsilon)$-approximation in $O(D \log n)$ rounds in the \LOCAL model, where ${D = \max_{(u,v) \in \mathcal S} \max_{P \in P_{u,v}} \ell(P)}$.  Moreover, if the convex program can be solved in polynomial time in the centralized sequential setting, then the distributed algorithm uses only polynomial-time computations at every node
\end{theorem}

The dependence on $\epsilon$ in the above theorem is hidden in the $O(\cdot)$ notation -- see Theorem~\ref{thm:distributed_LP} for the full statement.

Our main technique is to use a distributed construction of \emph{padded decompositions}, a specific type of network decomposition which we explain in detail in Section~\ref{decomposition-sec}.  Padded decompositions have been very useful for metric embeddings and approximation algorithms (e.g., \cite{GHR06,KLMN04}), but to the best of our knowledge have not been used before in distributed algorithms (with the exception of~\cite{DK2011}, which used a special case of them to give a distributed algorithm for the fault-tolerant $2$-spanner problem).  Very similar decompositions, such as the famous Linial-Saks decomposition~\cite{LS93}, have been used extensively in distributed settings, but the guarantees for padded decompositions are somewhat different (and we believe that these decompositions may prove useful in the future when designing distributed approximation algorithms).  In Section~\ref{decomposition-sec} we give a distributed algorithm in the \LOCAL model to construct padded decompositions.  These padded decompositions allow us to solve a collection of ``local'' convex programs with the guarantees that a) most of the demands in $\mathcal S$ are satisfied in one of the local programs, and b) the solutions of the local convex programs combine into a (possibly infeasible) global solution with cost at most the cost of the global optimum.  Then by averaging over $O(\log n)$ of these decompositions we get a feasible global solution which is almost optimal.

\subsubsection{Distributed approximation algorithms for network design}

Solving convex programming problems in distributed environments is interesting in its own right, and Theorem~\ref{thm:main-informal} is our main technical contribution, but the particular class of convex programs that we can solve are mostly interesting as convex relaxations of interesting combinatorial optimization problems.  Many of the problems are NP-hard, but there has been significant work (some quite recent) on designing approximation algorithms for them (see, e.g., \cite{DK-STOC,CD16,DZ16,Berman2011}).  Almost all of these approximations depend on convex relaxations which fall into our class of ``distance-bounded network design convex programs''.  This means that as long as the rounding scheme can be computed locally, we can design distributed versions of these approximation algorithms by using Theorem~\ref{thm:main-informal} to solve the appropriate convex relaxation and then using the local rounding scheme.

We are able to use this framework to give distributed approximation algorithms for several problems. Most of them are variations of \emph{graph spanners}, which were introduced by Peleg and Ullman \cite{PelegU87} and Peleg and Sch{\"{a}}ffer \cite{PelegS89}, and are defined as follows.

\begin{definition}
Let $G = (V, E)$ be a graph (possibly directed), and let $k \in \mathbb{N}$.  A subgraph $H$ of $G$ is a \emph{$k$-spanner} of $G$ if $d_H(u,v) \leq k \cdot d_G(u,v)$ for all $u,v \in V$.  The value $k$ is called the \emph{stretch} of the spanner.
\end{definition}

Before stating our results, we first define the problems.  In the \BASIC problem we are given an undirected graph $G$ and a value $k \in \mathbb{N}$.  A subgraph $H$ of $G$ is a feasible solution if it is a $k$-spanner of $G$, and the objective is to minimize the number of edges in $H$.  For $k = 3,4$, the best-known approximation algorithm for this problem is $\tilde O(n^{1/3})$~\cite{Berman2011,DZ16}.  If the input graph $G$ (and the solution $H$) are directed, then this is the \DIRECTED problem, for which the best-known approximation is $\tilde O(\sqrt{n})$~\cite{Berman2011}.  If the objective is instead to minimize the maximum degree in $H$ then this is the \DEGREE problem, for which the best-known approximation is $\tilde O(n^{\left(1-1/k\right)^2})$~\cite{CD16}.

The following theorem contains our results on distributed approximations of graph spanners.  Informally, it states that for we can give the same approximations in the the \LOCAL model as are possible in the centralized model.  

\begin{theorem} \label{thm:spanner-main}
There are algorithms in the \LOCAL model with the following guarantees.  For \DIRECTED, the algorithm runs in $O(k \log n)$ rounds and gives an $\tilde O(\sqrt{n})$-approximation.  For \BASIC[3] and \BASIC[4], the algorithms run in $O(\log n)$ rounds and gives an $\tilde O(n^{1/3})$-approximation.  For \DEGREE, the algorithm runs in $O(k \log n)$ rounds and gives an $\tilde O(n^{\left(1-1/k\right)^2})$-approximation.  All of these algorithms use only polynomial-time computations at each node.
\end{theorem}

We emphasize that our algorithms for these spanner problems both match the best-known centralized approximations and only use polynomial-time computations at each node.  There is significant previous work (see Section~\ref{sec:related}) on designing distributed approximation algorithms for these and related problems that has only one of these two properties, but all previous approaches which use only polynomial-time computations necessarily do worse than the best centralized bound (or have much worse round complexity).  At a high level, this is because previous approaches (most notably~\cite{BEC2016}) do not actually use the structure of the centralized algorithm: they only use the efficient centralized approximation as a black box.  By going inside the black box and noticing that they all use a similar type of convex relaxation, we can simultaneously get low round complexity, best-known approximation ratios, and efficient local computation.  

It turns out that we can use our techniques for an even broader question: \DSNDISTANCE with a spanning demand graph.  In this problem there is a set $\mathcal S \subseteq V \times V$ of demands, and for every demand $(u,v) \in \mathcal S$ there is a length bound $L(u,v)$.  The goal is to find a subgraph $H$ so that $d_H(u,v) \leq L(u,v)$ for all $(u,v) \in \mathcal S$, and the objective is to minimize the number of edges in $H$.  The state of the art centralized bound for this problem is a $O(n^{3/5 + \epsilon})$-approximation~\cite{CDKL17}, but if we further assume that every vertex $u \in V$ is the endpoint of at least one demand in $\mathcal S$ (which we will refer to as a \emph{spanning demand graph}) then it is straightforward to see that the centralized algorithm of~\cite{Berman2011} for \DIRECTED can be generalized to give a $\tilde O(\sqrt{n})$-approximation. Our distributed version of this algorithm also generalizes, giving the following result.

\begin{theorem} \label{thm:DSN-main}
There is an approximation algorithm in the \LOCAL model for \DSNDISTANCE with a spanning demand graph with approximation ratio $\tilde O(\sqrt{n})$ which runs in $O((\max_{(u,v) \in \mathcal S} L(u,v)) \log n)$ rounds and uses only polynomial-time computations.
\end{theorem}

Note that \DIRECTED and \BASIC are special cases of this problem, where there is a demand for every edge and the length bound is just $k$ times the original distance.  Interestingly, other network design problems which have proved important for distributed systems are also special cases, including the \textsc{Distance Preserver} problem (when $L(u,v) = d_G(u,v)$ for all $(u,v) \in \mathcal S$), the \textsc{Pairwise $k$-Spanner} problem (where $L(u,v) = k \cdot d_G(u,v)$ for all $(u,v) \in \mathcal S$), and the \SLSN problem (where $L(u,v) = D$ for all $(u,v) \in \mathcal S$, for some global parameter $D$).  \SLSN in particular is a key component in state of the art systems for reliable Internet transport~\cite{BWDA17}, although in that particular application the demand graph is not spanning.  Extending our techniques to handle totally general demands by giving a distributed version of~\cite{CDKL17} is an extremely interesting open question.

\subsection{Related Work} \label{sec:related}

While distributed solving of convex programs is a natural question, there is little previous work in the \LOCAL model.  Possibly most related to our results is a line of work on solving \emph{positive} linear programs (packing and covering LPs).  This was introduced by~\cite{PY93}, improved by~\cite{BBR97}, and then essentially optimal upper and lower bounds were given by~\cite{KMW06}.  Unfortunately, the convex programs we consider are not positive linear programs due to the ``capacity'' constraints in which some variables appear with positive coefficients while others have negative coefficients.  

A special case of our result was proved earlier in~\cite{DK2011}, who showed how to solve the LP relaxation of \BASIC[2] in the \LOCAL model in $O(\log^2 n)$ rounds (they actually show more than this, by giving a distributed algorithm for the \emph{fault-tolerant} version of \BASIC[2], but that is not germane to our results).  Our techniques are heavily based on~\cite{DK2011}, which is itself based on the ideas from~\cite{KMW06}.  In particular, \cite{KMW06} uses a Linial-Saks decomposition~\cite{LS93} to solve ``local'' versions of the linear program in different parts of the graph, and then combines these appropriately.  To make this work for the \BASIC[2] LP relaxation, \cite{DK2011} had to use \emph{padded decompositions}, which can be thought of as a variant of Linial-Saks with slightly different guarantees which, for technical reasons, are more useful for network design LPs.  In this paper we extend these techniques further by giving a more general definition of padded decomposition which works for larger distance requirements, showing how to construct them in the \LOCAL model, and then showing that the basic ``combining'' idea from~\cite{DK2011} can be extended to handle these more general decompositions and far more general constraints and objective functions.  

The major type of combinatorial optimization problem which our techniques allow us to approximate are various versions of graph spanners.   There are an enormous number of papers on spanners in both centralized and distributed models, but fewer papers which attempt to find the ``best'' spanner for the particular given input graphs (most papers on spanners give existential results and algorithms to achieve them, rather than optimization results).  These optimization questions (e.g., \BASIC, \DIRECTED, and \DEGREE) have been considered quite a bit in the context of centralized approximation algorithms and hardness of approximation~\cite{DK-STOC,Berman2011,DZ16,CD16,DKR16}, but almost all of the known centralized results use linear programming relaxations, making them difficult to adapt to distributed settings.  Hence there have been only two results on optimization bounds in distributed models: \cite{DK2011} and \cite{BEC2016}.

Barenboim et al.~\cite{BEC2016} provided a distributed algorithm using Linial-Saks decompositions that for any integer parameters $k, \alpha$, gives an $O(n^{1/\alpha})$-approximation for \DIRECTED in $\exp(O(\alpha)) +O(k)$ time. This is an extremely strong approximation bound, and in fact is  better than even the best centralized bound.  This is possible due to their use of very heavy (exponential time) local computation.  Our algorithms, on the other hand, take polynomial time for local computations.  Barenboim et al.~\cite{BEC2016} show that heavy local computations can be removed from their algorithm by using a centralized approximation algorithm for a variant of spanners known as \emph{client-server $k$-spanners}, and in particular that an $f(k)$-approximation for client-server $k$-spanner can be turned into an $O(n^{1/\alpha} f(k))$-approximation algorithm running in $\exp(O(\alpha)) + O(k)$ rounds in the $\mathcal{LOCAL}$ model for minimum $k$-spanner with only polynomial local computation.  So in order to achieve the same asymptotic approximation ratio as the best-known centralized algorithm, the parameter $\alpha$ must be $\Omega(\log n)$ and hence the running time is polynomial in $n$, even though $k$ might be a constant.  It is essentially known (though not written anywhere) that a variety of other results with slightly different tradeoffs can be achieved through similar uses of Linial-Saks~\cite{Elkin17} or refinements of Linial-Saks such as~\cite{EN16}.  However, since all of these approaches treat the centralized approximation algorithm as a black box, none of them can achieve the same approximation ratio as the centralized algorithm without suffering a much worse (usually polynomial) round complexity that the $O(k \log n)$ that we achieve.

\section{Preliminaries and Notation}
The distributed setting we will be considering is the \LOCAL model~\cite{Peleg00}, in which time passes in synchronous rounds and in each round every node can send an arbitrary message of unbounded size to each of its neighbors in the underlying graph $G = (V, E)$ (as always, we will let $n = |V|$ and $m = |E|$). We will assume that all nodes know $n$ (or at least know a constant approximation of $n$).  Usually in this model the communication graph is the same as the graph of computational interest; e.g., we will be trying to compute a spanner of the communication graph itself.  But for some applications we will want the graph to be directed, in which case we make the standard assumption that communiocation is bidirectional: the graph for which we are trying to compute a convex relaxation / network design problem is directed, but messages can be sent in both directions across a link.  In other words, the communication graph is just the undirected version of the given directed graph.  

For any pair of nodes $u,v \in V$ we define $d(u,v)$ to be the distance between $u$ and $v$ in the communication graph (i.e.~the length of a shortest path between $u$ and $v$ regardless of edge directions). We define $B(u,k)$ to be an undirected ball of radius $k$ from $u$ in the communication graph. More precisely, $B(u,k) =\{w \in V \mid d(u,w) \leq k\}$. 
	

If $x$ is a vector then we use $x_i$ to denote the $i$'th component of $x$.  Most of the time our vectors will be indexed by edges in a graph, in which case we will also use the notation $(x_e)_{e \in E}$.

Given a partition of the vertices $V$ of a graph, we will refer to each part of the partition as a ``cluster''. For any graph $G=(V,E)$ and set $S \subseteq V$, we let $E(S)$ denote the set of edges in the subgraph induced by $S$, i.e., $E(S)= \{ (u,v) \in E \mid u , v \in S\}$. We will frequently need ``restrictions'' of vectors to induced subgraphs, so for any vector $x \in \mathbb{R}^m$, we define $x^{S}=(x_e^{S})_{e \in E}$ to be the vector in $\mathbb{R}^m$ where $x^{S}_e =0$ if $e \not\in E(S)$ and  $x^{S}_e=x_e$ if $e \in E(S)$.

\section{Padded decompositions}\label{decomposition-sec}
We will now define and give an algorithm to construct \emph{padded decompositions}, which are one of the key technical tools that we will use when designing algorithm to solve distance-bounded network design convex programs.  In this section all graphs are undirected and all distances are with respect to this undirected graphs (in fact, the definition and our algorithm work more generally for any metric space). Recall that $B(u,k)$ denotes the undirected ball of radius $k$ from node $u$ (in the communication graph). 

\begin{definition}\label{def:padded_decomposition}
Given an \textit{undirected} graph $G$, a $(k, \epsilon)$-padded decomposition, where $0 <\epsilon \leq 1$, is a probability measure $\mu$ over the set of graph partitions (clusterings) that has the following properties:
\begin{itemize}
\item[1)] For every $P \in \text{supp}(\mu)$, and every cluster $C \in P$, we have: $diam(C) \leq O( (k/\epsilon) \log n)$.
\item[2)] For every $u \in V$, it holds that $\Pr( \exists C \in P \mid B(u,k) \subseteq C) \geq 1-\epsilon$. That is to say, the probability that all nodes in $B(u,k)$ are in the same cluster is at least $1-\epsilon$.
\end{itemize} 
\end{definition}

This notion of padded decompositions is standard in metric embeddings and approximation algorithms~\cite{GHR06,KLMN04}, but to the best of our knowledge has not yet been used in distributed algorithms.   We first use a centralized algorithm (Algorithm \ref{alg:padded_decomposition}) to sample from a $(k, \epsilon)$-padded decomposition, and then describe how it can be implemented in the $\mathcal{LOCAL}$ model. Algorithm \ref{alg:padded_decomposition} and its analysis are similar to a partitioning algorithm proposed in \cite{Bartal96}, which was shown to have a low probability of separating nodes in a close neighborhood.  

\begin{algorithm}[h]
\caption{Sampling from a $(k, \epsilon)$-padded decomposition of $G=(V,E)$.}
\label{alg:padded_decomposition}
Let $\pi: V \rightarrow [n]$ be an arbitrary bijection from $V$ to $[n]$, and let $r= (\frac{2}{\epsilon}) k$.\\
\For{$v \in V$}{
  Sample $z_v$ independently from a distribution with probability density function $p(z_v)= \left(\frac{n}{n-1}\right)\frac{e^{-{z_v}/r}}{r}$.\\
  Set the radius ${r_v= \min(z_v, r \ln n +k)}$.
  }
  \For{$u \in V$}{
   Node $u$ joins cluster $C(v)$, such that $d(v,u) \leq r_v \land (\pi(v) < \pi(w) \  \forall w\neq v \text{ s.t. }d(w,u) \leq r_w)$. 
    \tcp{Node $u$ joins the cluster $C(v)$, with cluster center $v$, which is the first node in the permutation where $d(v,u) \leq r_v$.}
  }
\end{algorithm}

For any partition $P$ constructed by Algorithm \ref{alg:padded_decomposition}, each cluster is clearly $C(v)$ for some $v \in V$.  We call this special node $v$ the \emph{center} of cluster $C(v)$. Later, we will use the center of each cluster for solving locally defined convex programs.
\begin{lemma} \label{ball-preserve}
Algorithm \ref{alg:padded_decomposition} partitions a given undirected graph $G=(V,E)$ into a partition $P$ such that $P$ is sampled from a $(k, \epsilon)$-padded decomposition.
\end{lemma}
\iflipics
\else
\begin{proof}
The first property in Definition \ref{def:padded_decomposition} is directly implied by the definition of $r_v$ for all nodes $v \in V$. 

 For the second property we consider an arbitrary node $u \in V$, and compute the probability that the ball $B(u,k)$ is not in any of the clusters in $P$. Consider an arbitrary value $1 \leq t \leq n$, let $v \in V$ be the node such that $t=\pi(v)$, and let $z=z_v$ be the real number sampled by $v$. 
Also, for any $x,y \in V$, let ${\tilde{d}(x,y)= \min(d(x,y), r \ln n +k)}$. Let us also order the clusters based on their center's position in the permutation, so that $C_t$ is the cluster corresponding to $t=\pi(v)$ (i.e.~$v$ is the cluster center of $C_t$). We define $X_t$ to be the event that if $B(u,k)$ is not in the first $t-1$ clusters, then it is also not in any of the remaining clusters. We provide a recursive bound on $X_t$ based on $X_{t+1}$. Then we will get the second property once we show $\Pr(X_0) \leq \epsilon$.
 We need to define the following events:
\begin{itemize}
\item $A_t: B(u,k)$ does not intersect with any of the clusters $C_1,..,C_{t-1}$.
\item $M^{cut}_t: (\tilde{d}(v,u)-k \leq z < \tilde{d}(v,u)+k \mid A_t)$.
\item $M^{ex}_t: (z<\tilde{d}(v,u)-k \mid A_t)$.
\item $X_t: (\nexists j \geq t: B(u,k) \subseteq C_j \mid A_t)$.
\end{itemize}
In other words, conditional on the event that $B(u,k)$ is not in any of the first $t-1$ clusters, either $B(u,k) \subseteq C_t$, or else one the following two events will occur: $M^{cut}_t$ is the event that $B(u,k)$ partially intersects $C_t$, and $M^{ex}$ is the event that $B(u,k)$ does not intersect $C_t$.
Now the event $X_t$ occurs only when either $M^{cut}_t$ occurs or both $M^{ex}_t$ and $X_{t+1}$ occur (i.e.~when $B(u,k)$ is not in $C_t$ or any of the next clusters). Hence we can write $\Pr(X_t) \leq \Pr(M^{cut}_t) + \Pr(M^{ex}_t)\Pr(X_{t+1})$.
Recall that $z$ is independently sampled from the density function $p(z_v)= \left(\frac{n}{n-1}\right)\frac{e^{-{z_v}/r}}{r}$, and thus $M^{cut}$ can be written as follows:
\begin{align*}
\Pr(M^{cut}_t) &= \int_{\tilde{d}(v,u)-k}^{\tilde{d}(v,u)+k} p(z)d_z
			   = \left(\frac{n}{n-1}\right)\left(1-e^{-2k/r}\right)e^{-(\tilde{d}(v,u)-k)/r}
			   \leq \left(\frac{n}{n-1}\right)\frac{2k}{r} e^{-(\tilde{d}(v,u)-k)/r}.
\end{align*}
Similarly, we can write,
\begin{align*}
\Pr(M^{ex}_t) &= \int_{0}^{\tilde{d}(v,u)-k} p(z)d_z
			   = \left(\frac{n}{n-1}\right)\left(1-e^{-(\tilde{d}(v,u)-k)/r}\right).
\end{align*}
We now inductively prove that $\Pr(X_t) \leq (2- \frac{t}{n-1})(\frac{2k}{r})$. If $t < n$ is the last step, then $\Pr(X_t)=0$, and thus this bound clearly holds. Assume that the bound is true for $X_{t+1}$, we show that then it also holds for $X_t$. We have,
\begin{align*}
\Pr(X_t) &\leq \Pr(M^{cut}_t) + \Pr(M^{ex}_t)\Pr(X_{t+1}) \leq \left(\frac{n}{n-1}\right)\left(\frac{2k}{r}\right)\left(1+ \frac{n-t-2}{n-1}\left(1-e^{-(\tilde{d}(v,u)-k)/r}\right)\right).
\end{align*}
Since $e^{-(\tilde{d}(v,u)-k)/r} \geq e^{-(\ln n)} \geq 1/n$, we get that
$\Pr(X_t) \leq \left(2- \frac{t}{n-1}\right)\left(\frac{2k}{r}\right)$. The second property is then implied by the fact that $\Pr(X_0) \leq \frac{2k}{r}= \frac{2k}{2k(1/\epsilon)} = \epsilon$.
\end{proof}
\fi

We will now use an idea similar to the one used in~\cite{DK2011} to make Algorithm~\ref{alg:padded_decomposition} distributed. In \cite{DK2011} they only considered the special case of $k=1$ and $\epsilon = 1/2$, which is why we cannot simply use their result as a black box.
\begin{lemma} \label{lem:decomp-distributed}
There is an algorithm in the $\mathcal{LOCAL}$ model that runs in $O(\frac{k}{\epsilon} \ln n)$ rounds and samples from a $(k, \epsilon)$-padded decomposition (so every node knows the cluster that it is in).
\end{lemma}
\begin{proof}
Without loss of generality, we assume that all nodes have unique IDs\footnote{We can make this assumption since nodes can each draw an ID from a suitably large space, so the probability of a collision is small enough that it does not affect the guarantees required by a $(k, \epsilon)$-padded decomposition. In our model, we assume that nodes know the size of the network}. The sequence of IDs in ascending order will determine the permutation $\pi$ used in Algorithm \ref{alg:padded_decomposition}, i.e.~if $\text{ID}_u < \text{ID}_v$ then $\pi(u) < \pi(v)$. The algorithm proceeds as follows until all nodes have been assigned to a cluster: each node $u \in V$ chooses a radius $r_u$ based on the distribution defined in Algorithm \ref{alg:padded_decomposition}. Then every $u \in V$ simultaneously sends a message containing ID$_u$ to all nodes in $B(u,r_u)$. After receiving all the messages, each node chooses the node with the smallest ID as the cluster center. Then Lemma \ref{ball-preserve} implies that the clusters satisfy the properties of a $(k, \epsilon)$-padded decomposition. Since the radius that each node chooses is $O((k/\epsilon) \log n)$, and each node only communicates with nodes within its radius, the running time in the $\mathcal{LOCAL}$ mode is $O((k/\epsilon) \log n)$.
\end{proof}

\section{Distributed distance bounded network design convex programming} \label{sec:solving-convex}
 In this section we prove Theorem~\ref{thm:main-informal}, giving an algorithm similar to~\cite{DK2011} which can almost optimally solve distance-bounded network design convex programs.  We first make all definitions formal in Section~\ref{sec:program-def}, and in particular define formally the class of objective functions where our results hold.  Then in Section~\ref{sec:CP-distributed} we give a distributed algorithm which solved these programs up to arbitrarily small error.  
%
  
\subsection{Distance bounded network design convex programs} \label{sec:program-def}
 We will first describe a general class of objective functions that our algorithm applies to.
  For a graph $G=(V,E)$ and a set $S \subseteq V$, we let $E(S)$ denote the set of edges in the subgraph of $G$ induced by $S$. Recall that for a vector $x\in \mathbb{R}^m$ (where $m=|E|$), we define $x^{S}=(x_e^{S})_{e \in E} \in \mathbb{R}^m$ to be the vector where $x^{S}_e =0$ if $e \not\in E(S)$ and  $x^{S}_e=x_e$ if $e \in E(S)$. 
\begin{definition} \label{def_partition_func}
  Given a graph $G=(V,E)$, a function $g: \mathbb{R}^m \mapsto \mathbb{R}$ is \textit{convex partitionable} with respect to $G$ if $g$ is a non-decreasing\footnote{Let $f(x_1,...,x_k)$ be a multivariate function. We will say $f$ is nondecreasing if the following holds: if $x_i \leq x'_i$ for all $1 \leq i \leq k$, then $f(x_1,...,x_k) \leq f(x'_1,...,x'_k)$.} and convex  function with the following property:~for all partitions $\sigma = \{\sigma_1,...,\sigma_\ell\}$ of nodes in $V$, there exists a non-decreasing function $h_\sigma: \mathbb{R}^\ell \mapsto \mathbb{R}$, s.t.~${g(x) = h_\sigma(g(x^{\sigma_1}), g(x^{\sigma_2}),...,g(x^{\sigma_\ell}))}$ for all $x=(x_e)_{e \in E}$ where $x_e=0$ for any edge $e$ with endpoints in different clusters of $\sigma$ (equality does not need to hold for vectors $x$ with nonzero values on edges between clusters).
\end{definition}
 Convex partitionable functions for graphs are a natural class of functions for distributed computing purposes. Moreover, this class includes many types of objective functions that are of interest in network design problems, including $p$-norms and linear functions. For example, if the function $g$ is the $p$-norm with $p \in \mathbb{Z}_{\geq 0}$, then it is easy to verify that by setting the function $h_\sigma$ to also be the $p$-norm for any partition $\sigma$ of $V$, the conditions of Definition \ref{def_partition_func} are satisfied. Note that an unweighted sum is the $1$-norm, and the $\max$ function is the infinity norm, and hence they will also satisfy the conditions of Definition \ref{def_partition_func}. Similarly, in case of linear functions, it is easy to see that conditions of Definition \ref{def_partition_func} are satisfied by setting $h_\sigma$ to be the \emph{unweighted} sum.
 
There are also other, less trivial examples.  For example, it is not hard to show the $p$-norm of the \emph{degree vector} (rather than just the edge vector) is also convex partitionable with respect to $G$. An important special case of this is the $\infty$-norm of the degree vector, i.e., the maximum degree. Since we use this objective in some of our applications (i.e., for the \DEGREE problem), we give a short proof of this case in \iflipics\else Lemma~\ref{lem:max-deg-partition}\fi.  For an integral vector $x \in \mathbb{R}^m$ we can write ${g(x)=\max_{v \in V} \deg(v)}$. By generalizing this notation to all $x \in \mathbb{R}^m$, we can define fractional node degrees as ${\deg(v)=\sum_{u:(v,u) \in E} x_{(v,u)}}$\footnote{Here we are considering out-degree of nodes in a directed graph. It is easy to see that Lemma \ref{lem:max-deg-partition} also holds in cases of in-degree only or sum of out-degree and in-degree. The later is the case we are interested for Section \ref{sec:applications}.}.
\begin{lemma} \label{lem:max-deg-partition}
 Given a graph $G=(V,E)$, the function $g(x)= \max_{v \in V} (\sum_{u:(v,u) \in E} x_{(v,u)})$ is convex partitionable w.r.t.~$G$. 
\end{lemma} 
\iflipics
\else
\begin{proof} 
Let $\sigma = \{\sigma_1,...,\sigma_\ell\}$ be a partition of nodes in $V$. For all $1 \leq i \leq \ell$, we have ${g(x^{\sigma_i})=\max_{v \in \sigma_i}(\sum_{u:(v,u) \in E} x^{\sigma_i}_{(v,u)})}$. Then we can set $h_\sigma (y) =\max_{i \in [\ell]} (y_i), y \in \mathbb{R}^\ell$, where $y_i$ is the $i$-th coordinate of $y$. Let $\sigma(v) \in \sigma$ be the cluster that node $v$ belongs to. For all $x=(x_{(u,v)})_{(u,v) \in E}$, where  $x_{(u,v)}=0$ for any $(u,v) \in E$ s.t.~$\sigma(u) \neq \sigma(v)$ , we have,
\begin{align*}
g(x)&= \max_{v \in V} \left(\sum_{u:(v,u) \in E} x_{(v,u)} \right)= \max_{\sigma_i \in \sigma} \left( \max_{v \in \sigma_i} \left(\sum_{u:(v,u) \in E} x^{\sigma_i}_{(v,u)} \right)\right)\\
&=  \max_{\sigma_i \in \sigma} \left( g\left(x^{\sigma_i}\right) \right)= h_\sigma(g(x^{\sigma_1}), g(x^{\sigma_2}),..., g(x^{\sigma_\ell})). 
\end{align*}
  It is also easy to see that the function $h_\sigma$ is convex and non-decreasing. Hence $h_\sigma$ satisfies the conditions in Definition \ref{def_partition_func}.
\end{proof}
\fi
%

Now that this class of functions has been defined, we can formally define the class of distance-bounded network design convex programs.   
\begin{definition}\label{general-LP}
 Let $ \mathcal{S} \subseteq V\times V$ be a set of pairs in the graph $G=(V,E)$, and for any pair $(u,v) \in \mathcal{S}$ let $\mathcal{P}_{u,v}$ be a set of paths from $u$ to $v$, which we sometimes call the set of ``allowed'' paths. Let $g$ be a non-decreasing convex-partitionable function of ${x}=(x_e)_{e \in E}$ with $g(\vec{0})=0$.  Then we call a convex program of the following form a \textit{distance bounded network design CP}: 
 \begin{align*} 
\min \quad &g({x})\\
s.t &\sum_{P \in \mathcal{P}_{u,v}: e\in P} f_P\leq x_e  &\forall(u,v) \in  \mathcal{S}, \forall e \in E \\
&\sum_{P \in \mathcal{P}_{u,v}} f_P \geq 1 &\forall(u,v) \in  \mathcal{S}\\
&x_e \geq 0 &\forall	e \in E\\
&f_P \geq 0 &\forall(u,v)\in  \mathcal{S}, \forall P\in \mathcal{P}_{u,v}
\end{align*}
\end{definition}
As we will see in Section \ref{sec:applications}, many network design problems use linear (or convex) programming relaxations that satisfy the conditions of Definition \ref{general-LP}.  A key parameter of such a program is the length of the longest allowed path $D= \max_{(u,v) \in \mathcal{S}} \max_{p \in \mathcal{P}_{u,v}} \ell(p)$ (where $\ell(p)$ is the length of path $p$).

\subsection{Distributed Algorithm} \label{sec:CP-distributed}
In order to solve these convex programs in a distributed manner, we will first use padded decompositions to form a local problem using a simple distributed algorithm. Let $P$ be a partition sampled from a $(k, \epsilon)$-padded decomposition (in particular, obtained by Lemma~\ref{lem:decomp-distributed}), where $0 < \epsilon \leq 1$. Recall that for each cluster $C \in P$, $E(C)= \{ (u,v) \in E \mid u , v \in C\}$.  We define $G(C)$ to be the subgraph induced by $C$. 
\begin{lemma}\label{lem:cluster_centers}
For each cluster $C$ sampled from a $(k, \epsilon)$-padded decomposition, there is a distributed algorithm running in $O(\frac{k}{\epsilon} \log n)$ rounds so that every cluster center knows $G(C)$. 
\end{lemma}
\iflipics
\else
\begin{proof}
 The first property of $(k, \epsilon)$-padded decompositions implies that for all nodes $u \in C$, we have $d(u,v)=O((k/\epsilon) \log n)$, where $v$ is the center of cluster $C$. Each node $u \in C$ that determines $v$ as the center of the cluster it belongs to, will send the information of its incident edges to $v$. Since there is no bound on the size of the messages being forwarded, this can be done in $O((k/\epsilon) \log n)$ time. 
\end{proof}
\fi

 Let CP($G$) be a distance bounded network design CP defined on graph $G=(V,E)$. We will define local convex programs based on a partition $P$ of $G$ that is sampled from a $(D, \lambda)$-padded decomposition. The value of $0 < \lambda \leq 1$ will be set later based on the parameters of our distributed algorithm.  For each $C \in P$, let CP$(C)$ be CP$(G)$ defined on $G(C)$, but where only demands corresponding to any pair $(u,v) \in \mathcal{S}$ in which $B(u,D)$ is fully contained in $C$ are included. We denote the set of these demands by $N(C)$, more precisely, ${N(C)= \{ (u,v) \in \mathcal{S} \mid B(u,D) \subseteq C \}}$. The objective will then be to minimize $g(x)=g\left((x_e)_{e \in E(C)}\right)$.  In other words $CP(C)$ is defined as follows:
 \begin{align*}
\min \quad &g(x)\\
s.t &\sum_{P \in \mathcal{P}_{u,v}: e\in P} f_P\leq x_e  & \forall(u,v) \in N(C), \forall e \in E(C) \\
&\sum_{P \in \mathcal{P}_{u,v}} f_P \geq 1 & \forall(u,v) \in N(C)\\
&x_e \geq 0 &\forall e \in E(C)\\
&f_P \geq 0 &\forall(u,v)\in N(C), \forall P\in \mathcal{P}_{u,v}
\end{align*}

There is a technical subtlety about computing the function $g$ on each cluster, which is the fact that a solution $\langle x^C, f^C \rangle$ of CP($C$) is only defined on $G(C)$. While in practice $x^C$ is a vector defined only on edges in $E(C)$, in our analysis we will assume that $x^C$ is a vector in $\mathbb{R}^m$ (where $|E|=m$) and $x_e^C=0$ for all $e \not \in E(C)$ (this is possible for all the functions we care about). It will become clear that this distinction is also needed when we try to compare the local solutions to the global solutions.
 The following lemma is similar to Lemma 3.8 in~\cite{DK2011}, and we show that it holds for our modified definition of local convex programs and for generalized objective functions that satisfy Definition \ref{def_partition_func}.
\begin{lemma}\label{lem:cluster_upper}
Let $\langle x^*, f^*\rangle$ be an optimal solution of $CP(G)$ and let ${x^*}^C=(x^*_e)_{e \in E(C)}$. For each cluster $C \in P$, let $\langle \tilde{x}^C, \tilde{f}^C \rangle$ be an optimal solution of CP($C$). Then $g(\tilde{x}^C) \leq g({x^*}^C)$. 
\end{lemma}
\begin{proof}
 We argue that the vector $\langle {x^*}^C, {f^*}^C \rangle$, where ${x_e^*}^C= x^*_e$ for all $e \in E(C)$ and ${f_p^*}^C=f^*_p$ for all $p \in \mathcal{P}_{u,v}$, is a feasible solution to CP$(C)$.
By definition of $N(C)$ we have that for any $(u,v) \in N(C)$ all paths in $\mathcal{P}_{u,v}$ also appear in $G(C)$, and therefore $\langle {x^*}^C, {f^*}^C \rangle$ satisfies both capacity and flow constraints of CP$(C)$ for pairs $(u,v) \in E(C)$ since they were satisfied in CP$(G)$. Since we assumed that $\langle \tilde{x}^C, \tilde{f}^C \rangle$ is an optimal solution of CP($C$), this implies that $g(\tilde{x}^C) \leq g({x^*}^C)$.
\end{proof}

We now provide in Algorithm~\ref{alg:distributed_generalLP} a distributed algorithm for solving CP$(G)$, by having cluster centers solve CP$(C)$ of their cluster using a sequential algorithm in each iteration, and then averaging over the solutions for each edge.  We assume that all nodes know the values of $D$ and $\epsilon$. Let $C_{u,i}$ denote the cluster that node $u$ belongs to in the $i$-th iteration, and let $\langle x^{C_{u,i},i}, f^{C_{u,i},i} \rangle$ be the fractional CP solution of $C_{u,i}$, where $\langle x_e^{C_{u,i},i}, f_p^{C_{u,i},i} \rangle$ is the fractional CP value for $e=(u,v)$, and $p \in \mathcal{P}_{u,v}$.

Since the objective is a function of edge vectors, what we mean by having a distributed solution to a distance bounded network design CP is that each node $u$ will know the value $x_e$ for all the edges $e$ incident to $u$. It is not hard to see that the algorithm could be modified so that every node $u$ can also know the flow value $f_p$ for each path $p$.
\begin{algorithm}[h]
\caption{Distributed algorithm for approximating distance bounded network design CPs.}
\label{alg:distributed_generalLP}
  Set $\lambda= \frac{\epsilon (1 - \epsilon)}{(2-\epsilon)(1+ \epsilon)}$ and $t = \left\lceil  \frac{ 16 (1- \frac{\epsilon}{2}) (1+\epsilon)\ln n}{\epsilon^2} \right\rceil$.\\
  Sample from $(D, \lambda)$-padded decompositions $t$ times by Lemma \ref{lem:decomp-distributed}, and let $P_i$ be the partition obtained in the $i$'th run.\\
  For each cluster $C \in P_i$, the center of cluster $C$ computes $G(C)$ (see Lemma \ref{lem:cluster_centers}).\\
  The center of each cluster $C \in P_i$ solves $CP(C)$
  and sends the solution $\langle x^{C,i}, f^{C,i} \rangle$ to all nodes $u \in C$.

\For{$e=(u,v)\in E$}{
Let $I_{u,v}=\{i \mid \exists C \in P_i: u,v \in C \}$.\\
\tcp{these are the iterations in which both endpoints are in same cluster}
$\tilde x_e \leftarrow \min(1, \frac{1+\epsilon}{t}\sum_{i \in  I_{u,v}} x^{C_{u,i},i}_e)$.

}
\end{algorithm}

\begin{theorem} \label{thm:distributed_LP}
Algorithm \ref{alg:distributed_generalLP} takes $O((D/\epsilon)\log n)$ rounds to terminate, and it will compute a solution of cost $(1+\epsilon) CP^*$ to a bounded distance network design CP (Definition \ref{general-LP}) with high probability, where $CP^*$ is the optimal solution and $0 < \epsilon \leq 1$. Moreover, if the convex program can be solved sequentially to arbitrary precision in polynomial time, then all of the node computations are also polynomial time.
\end{theorem}
\begin{proof}
\textbf{Correctness:} We first show that with high probability the values $\tilde{x}_e, e \in E$ form a feasible solution. Here we only need to show that a feasible solution for the flow values exist, and do not require nodes to compute these values.
 Let $I_u = \{i: \exists C \in P_i, B(u,D) \subseteq C \}$, i.e.~$I_u$ is the set of iterations in which $B(u,D)$ is contained in a cluster, and let $I_{u,v}$ be the set of iterations in which both $u$ and $v$ are in the same cluster. Since we need to implement Algorithm \ref{alg:distributed_generalLP} in a distributed manner, we use $I_{u,v}$ in our implementation, while the analysis is based on $I_u$. We can do so since by definition we have $I_u \subseteq I_{u,v}$, for any $(u,v) \in E$. 
 
 For any $p \in \mathcal{P}_{u,v}$, we set the flow values to be $\tilde{f}_p = \frac{1}{|I_u|} \sum_{i \in I_u} f_p^{C_{u,i},i}$. We will show that this gives a feasible flow. First we argue that enough flow is being sent. For all $(u,v) \in \mathcal{S}$, we have,
 \begin{align*}
  \sum_{p \in \mathcal{P}_{u,v}} \tilde{f}_p &=\sum_{p \in \mathcal{P}_{u,v}} \frac{1}{|I_u|} \sum_{i \in I_u} f_p^{C_{u,i},i} =  \frac{1}{|I_u|} \sum_{i \in I_u} \sum_{p \in \mathcal{P}_{u,v}} f_p^{C_{u,i},i} \geq  \frac{1}{|I_u|} \sum_{i \in I_u}  1 \geq 1.
  \end{align*}
 We have used the fact that for each $i \in I_u$ the solution corresponding to the CP of the cluster containing $u$ satisfies the constraint that $\sum_{p \in \mathcal{P}_{u,v}: e\in p} f_p^{C_{u,i},i} \geq 1$, because for each such $i$ we know that $(u,v) \in N(C)$.
 
   Next, we will argue that the capacity constraints are also satisfied. The second property of $(D, \lambda)$-padded decompositions implies that ${\Pr(i \in I_u) \geq 1- \lambda =1-\frac{\epsilon (1 - \epsilon)}{(2-\epsilon)(1+ \epsilon)}= \frac{1}{(1-\frac{\epsilon}{2}) (1+\epsilon)}}$ for each iteration $1 \leq i \leq t$. By linearity of expectations we have $E[|I_u|] \geq t (1-\lambda)$. Since each sampling is performed independently, by Chernoff bound for $\delta=\epsilon/2$, we get, 
\begin{align*}
  \Pr(|I_u| \leq t(1-\lambda)(1-\delta))&= \Pr \left(|I_u| \leq \frac{t (1-\frac{\epsilon}{2})}{(1-\frac{\epsilon}{2}) (1+\epsilon)}\right)\\
  &= \Pr\left(|I_u| \leq \frac{t}{(1+\epsilon)}\right) \leq e^{- \frac{(\epsilon/2)^2(1-\lambda) t}{2}} \leq e^{-2\ln n} = \frac{1}{n^2}.
\end{align*}     
     Hence by a union bound on all nodes we have that with high probability ${|I_u| > t/(1+\epsilon)}$. Therefore, for all $(u,v) \in \mathcal{S}, e \in E$, we have (w.h.p.),
    \begin{align*}
  \sum_{p \in \mathcal{P}_{u,v}: e\in p} \tilde{f}_p &=\sum_{p \in \mathcal{P}_{u,v}: e\in p} \frac{1}{|I_u|} \sum_{i \in I_u} f_p^{C_{u,i},i} =  \frac{1}{|I_u|} \sum_{i \in I_u} \sum_{p \in \mathcal{P}_{u,v}: e\in p} f_p^{C_{u,i},i} \leq  \frac{1}{|I_u|} \sum_{i \in I_u}  x_{e}^{C_{u,i},i} \\ 
   &\leq \min\left(1, \frac{1}{|I_u|} \sum_{i \in I_{u,v}}  x_{e}^{C_{u,i},i}\right)  \leq \min\left(1, \frac{1+\epsilon}{t} \sum_{i \in I_{u,v}}  x_{e}^{C_{u,i},i}\right)= \tilde{x}_e.
  \end{align*}

\textbf{Upper bound:} We will now show that the upper bound holds. Let $\langle x^*, f^* \rangle$ be an optimal solution to CP($G$). We have ${\tilde{x}_e= \min(1, \frac{1+\epsilon}{t} \sum_{i \in I_e} {x}_e^{C_{u,i},i})}$, and for each $e = (u,v)$ and $1 \leq i \leq t$, we set $\tilde{x}^i_e= {x}_e^{C_{u,i},i}$ if $i \in I_e$, and $\tilde{x}^i_e = 0$ otherwise.   
Note that $0<(1+\epsilon)/t <1$, and since $g$ is a convex function and $g(\vec{0})=0$, by Jensen's inequality we have $g \left(\frac{1+ \epsilon}{t}x \right) \leq \frac{1+ \epsilon}{t} g(x)$. Then for $\tilde{x}=(\tilde{x}_e)_{e \in E}$ we can write: 
\begin{align*}
g({\tilde{x}}) &= g \left(\left(\tilde{x}_e\right)_{e \in E}\right) \leq g\left(\frac{1+\epsilon}{t}\left(\sum_{i \in I_e} x_e^{C_{u,i},i}\right)_{e \in E}\right)\leq \frac{1+\epsilon}{t} g \left(\left(\sum_{i \in I_e} x_e^{C_{u,i},i}\right)_{e \in E}\right)\\
				 &\leq \frac{1+\epsilon}{t}  g \left(\left(\sum^t_{i=1}\tilde{x}^i_e \right)_{e \in E}\right)\leq \frac{1+\epsilon}{t} g\left(\sum_{i=1}^t \left(\tilde{x}^i_e\right)_{e \in E}\right) \leq \frac{1+\epsilon}{t} \sum_{i=1}^t g \left(\left(\tilde{x}^i_e \right)_{e \in E}\right).
\end{align*}
In the final inequality, since $g$ is convex, we used Jensen's inequality to take the sum out of the function. It is now enough to show that in each iteration $i$, it holds $g((\tilde{x}^i_e)_{e \in E})\leq g({x^*})$.  Let $\tilde{x}^i =( (\tilde{x}_e^i)_{e \in E})$, and let $P_i=\{ C_1, C_2, ..., C_\ell \}$ be the partition of $V$. Since $g$ is a convex partitionable function w.r.t.~$G$, there exists a nondecreasing and convex function $h: \mathbb{R}^m \mapsto \mathbb{R}$ for which we can write $g({\tilde{x}}^i)= h_{P_i}(g(\tilde{x}^{i,C_1}), g(\tilde{x}^{i,C_2}), ..., g(\tilde{x}^{i, C_\ell}))$, since $\tilde x^i_e = 0$ by definition for edges which go between clusters (for simplicity we are denoting $\tilde{x}^{i^{C_j}}$ by $\tilde{x}^{i,C_j}$). 

Recall that ${x^*}^C$ is the vector in which ${x^*}^C_e= x^*_e$ for all $e \in E(C)$ and ${x^*}^C_e=0$ otherwise. By Lemma \ref{lem:cluster_upper} we get that for all $C \in P_i$, $g(\tilde{x}^{i,C}) \leq g({x^*}^{C})$. Now we consider a vector $\hat{x}$, defined by setting $\hat{x}_e= x^*_e$ for all edge $e$ with both endpoints in the same cluster, and $\hat{x}_e=0$ otherwise. Since we assumed $h_{P_i}$ to be nondecreasing, we get,
\begin{align*}
g(\tilde{x}^i)&= h_{P_i}(g(\tilde{x}^{i,C_1}), g(\tilde{x}^{i,C_2}), ..., g(\tilde{x}^{i, C\ell}))\leq h_{P_i}(g({x^*}^{C_1}), g({x^*}^{C_2}), ..., g({x^*}^{C_\ell}))\\ 
&=h_{P_i}(g(\hat{x}^{C_1}, \hat{x}^{C_2}, ..., g(\hat{x}^{C_\ell})) =g(\hat{x}) \leq g(x^*).
\end{align*}		
For the last inequality we have used the fact that $g$ is non-decreasing, and that for all $e \in E$, $\hat{x}_e \leq x^*_e$ (since either $\hat{x}_e=x^*_e$ or $\hat{x}_e=0$). By plugging this into the above inequalities, we will get $g(\tilde{x}) \leq \frac{1+\epsilon}{t} \sum_{i=1}^t g({\tilde{x}}^i) \leq (1+\epsilon) g({x^*})$, which implies the claim that Algorithm \ref{alg:distributed_generalLP} gives a $(1+\epsilon)$-approximation to the optimal solution. 

\textbf{Time Complexity:} The decomposition step and sending the information within a cluster takes $ O( (D/\epsilon)\log n)$ rounds since the diameter of each cluster is $O( (D/ \lambda)\log n) = O( (D/\epsilon) \log n)$. Since each decomposition is independent, we can do all of them in parallel, so steps 1-4 of the algorithm only take $O( (D/\epsilon) \log n)$ rounds in total.  Clearly the rest of the algorithm can be done in a constant number of rounds. Hence in total w.h.p.~the algorithm will take $O((D/\epsilon) \log n)$ rounds.
\end{proof}

\section{Distributed Approximation Algorithms for Network Design} \label{sec:applications}
In this section, we will focus on several network design problems which can be approximated by first solving a convex relaxation using Algorithm \ref{alg:distributed_generalLP} and then locally rounding the solution. For that purpose, we will describe how each problem has a distance bounded network design CP relaxation (Definition \ref{general-LP}), and will then show that existing rounding schemes are local.  

\subsection{\DIRECTED} \label{sec:spanner}
Dinitz and Krauthgamer~\cite{DK-STOC} introduced a linear programming relaxation for \DIRECTED which is just a distance-bounded network design CP with demands pairs $\mathcal{S} = E$, allowed paths $\mathcal P_{u,v}$ which are the directed paths from $u$ to $v$ of length at most $k$, and objective function $g(x) = \sum_{e \in E} x_e$. They showed that this LP can be solved in polynomial time (approximately if $k$ is non-constant).  
%
We will denote this LP by $LP(G)$. Clearly, LP$(G)$ is a distance bounded network design CP with $D=k$. Hence, Theorem \ref{thm:distributed_LP} implies that we can use Algorithm \ref{alg:distributed_generalLP} to approximately solve this LP in $O(k \log n)$ rounds in the \LOCAL model. 

We now provide in Algorithm~\ref{alg:kspanner} a distributed rounding scheme that gives an $O(n^{1/2} \log n)$-approximation for \DIRECTED. This algorithm matches the best centralized approximation ratio known~\cite{Berman2011}, and is just the obvious distributed version of the algorithm proposed in~\cite{Berman2011}. The difference is that here we truncate the shortest-path trees at depth $k$ (as opposed to full shortest-path trees), and nodes choose whether to become a tree root independently (rather than chosen without replacement as in~\cite{Berman2011}.  

\begin{algorithm}[h]
\caption{Distributed rounding algorithm for $k$-spanner.}
\label{alg:kspanner}
\SetKwInOut{Input}{Input}
\Input{Graph $G=(V,E)$, fractional solution $\langle x,f\rangle$ to LP$(G)$.}
$E'= \emptyset, \forall v \in V: T_v^{in} = \emptyset, T_v^{out}=\emptyset$.\\
\For{$e \in E$}{
Add $e$ to $E'$ with probability $\min(n^{1/2} \cdot \ln n \cdot x_e, 1)$.
}
\For{$v \in V$}{ \tcp{Random tree sampling}
Choose $p$ uniformly at random from $[0,1]$.\\
\If{$p < \frac{3 \ln n}{\sqrt{n}}$}{
 $T_v^{in} \leftarrow$ shortest path in-arborescence rooted at $v$ truncated at depth $k$. \\
 $T_v^{out} \leftarrow$ shortest path out-arborescence rooted at $v$ truncated at depth $k$.\\
   }
}
Output $E' \cup (\cup_{v \in V} (T_v^{in} \cup T_v^{out}))$. \tcp{Each node knows its portion of the output.}
\end{algorithm}

The following lemma is essentially from~\cite{Berman2011}, with the proof requiring only slight technical changes due to the slightly different algorithms.  We sketch it for completeness.

\begin{lemma} \label{lem:directed-round}
Given a directed graph $G$, LP$(G)$ as defined, and a fractional solution $LP^*$ to $LP(G)$, the output of Algorithm \ref{alg:kspanner} has size $O(n^{1/2} \cdot (n +LP^*)\log n )$.
\end{lemma}
\iflipics \else
\begin{proof}
  Let $N_{s,t}$ be the subgraph of $G$ induced by the nodes on paths in $\mathcal P_{s,t}$. Edge $e \in E$ is called a \emph{thick} edge if $|N_{s,t}| \geq n^{1/2}$, and otherwise it is called a \emph{thin} edge. The set $E'$ in Algorithm \ref{alg:kspanner} satisfies the spanner property for all thin edges (as argued in \cite{Berman2011}), and the random tree sampling phase satisfies the spanner property for the thick edges. Each thick edge $(s,t)$ is spanned if at least one node in $N_{s,t}$ performs the random tree sampling. This probability is at least $1-(1- \frac{3\ln n}{n^{1/2}})^{n^{1/2}} \geq 1 - 1/{n^3}$. Then a union bound on all the edges (of size at most $O(n^2)$) implies that w.h.p.~all thick edges are spanned.
We now argue that the output is an $O(n^{1/2} \log n)$-approximation algorithm: at most $O(n^{1/2} \log n)$ arborescences are chosen with high probability (each arborscence has $O(n)$ edges), and we argued that $|E'|= O(n^{1/2} \log n \cdot LP^*)$. Hence, the overall size of the output is $O(n^{1/2} \log n \cdot (n+ LP^*))$. 
\end{proof}
\fi

It is also easy to see that this algorithm can be implemented in the $\mathcal{LOCAL}$ model.

\begin{lemma}\label{lem:distributed_rounding}
Algorithm \ref{alg:kspanner} runs in $O(k)$ time in the $\mathcal{LOCAL}$ model.
\end{lemma}
\iflipics \else
\begin{proof}
Each node $v$ in $G$ has received the fractional solutions $x_e$ corresponding to all edges $e \in E$ incident to $v$. The randomized rounding step can be performed locally: the node with the smaller ID flips a coin, and exchanges the coin flip result with its corresponding neighbors. In order to form $T_i^{in}$ and $T_i^{out}$, $v$ performs a distributed BFS algorithms by forming a shortest path tree while keeping track of the distance from $v$. When the distance counter reaches $k$, the tree construction terminates. 
\end{proof}
\fi

We now immediately get our main result for \DIRECTED.

\begin{corollary}
Algorithm \ref{alg:distributed_generalLP} with $D=k$ along with the rounding scheme in Algorithm \ref{alg:kspanner} yields an $O(n^{1/2} \ln n)$-approximation w.h.p.~to \DIRECTED that runs in $O(k \log n)$ time in the $\mathcal{LOCAL}$ model and uses only polynomial-time computations at each node.
\end{corollary}
\begin{proof}
We first run Algorithm \ref{alg:distributed_generalLP} to solve LP$(G)$ up to a constant factor (by setting $\epsilon = 1/2$), which takes time $O(k \log n)$ with high probability (Theorem \ref{thm:distributed_LP}). Since each cluster center can solve the local LP in polynomial time, all computations are polynomial time.  We then use Algorithm \ref{alg:kspanner} to round the fractional solutions of LP$(G)$, which takes $O(k)$ time. Since the size of a $k$-spanner is at least $\Omega(n)$, Algorithm \ref{alg:kspanner} then outputs an $O(n^{1/2} \ln n)$-approximation to the minimum (Lemma \ref{lem:directed-round}).
\end{proof}

\subsection{\BASIC[3] and \BASIC[4]}
If the input graph is undirected then stronger approximations are possible.  In particular, for stretch $3$ and $4$, there are $\tilde O(n^{1/3})$-approximations due to~\cite{Berman2011} (for stretch $3$) and \cite{DZ16} (for stretch $4$).  Without going into details, both of these algorithms use the same LP relaxation as in \DIRECTED, but round the LP differently.  So in order to give distributed versions of these algorithms, we only need to modify Algorithm~\ref{alg:kspanner} to use the appropriate rounding algorithm (and change some of the other parameters in the shortest-path arborescence sampling).  Fortunately, both of these algorithms use rounding schemes which are highly local.  Informally, rather than sample each edge independently with probability proportional to the (inflated) fractional value as in Algorithm~\ref{alg:kspanner}, these algorithms sample a value independently at each \emph{vertex} and then include an edge if a particular function of the values of the two endpoints (different in each of the algorithms) passes some threshold.  Clearly this is a very local rounding algorithm: once we have solved the LP relaxation using Theorem~\ref{thm:distributed_LP}, each node can draw its random value and then spend one more round to exchange a message with each of its neighbors to find out their values, and thus determine which of the edges have been included by the rounding.  Thus the total running time is dominated by the time needed to solve the LP, which in these cases is $O(\log n)$ using Theorem~\ref{thm:distributed_LP}.

\subsection{\DEGREE} \label{sec:LDkS}
We now turn our attention to \DEGREE: Given a graph $G=(V,E)$ and a value $k$, we want to find a $k$-spanner that minimizes the maximum degree. We will use the relaxation and rounding scheme proposed by Chlamt\'a\v{c} and Dinitz~\cite{CD16}. The linear programming relaxation used in \cite{CD16} is very similar to the Directed and Basic $k$-spanner LP relaxation described earlier, with the difference being that a new variable $\lambda$ is added to represent the maximum degree, and so the objective is to minimize $\lambda$ and constraints are added to force $\lambda$ to upper bound the maximum fractional degree.  
\begin{theorem}
Given a graph $G=(V,E)$ (directed or undirected), and any integer $k \geq 1$ there is a distributed algorithm that w.h.p.~computes an $\tilde{O}(\Delta^{(1-1/k)^2})$-approximation to the \DEGREE problem, taking $O(k \log n)$ rounds of the $\mathcal{LOCAL}$ model and using only polynomial-time computations at each node.
\end{theorem}
\begin{proof}
It is easy to see that the LP relaxation proposed in \cite{CD16} can be written as a distance bounded network design CP where the objective is $\max_{v \in V}{\deg(v)= \max_{v \in V}\sum_{u:\{v,u\} \in E} x_{\{v,u\}}}$ (we do not need to use their extra variable $\lambda$, since we can instead directly write the objective). Lemma \ref{lem:max-deg-partition} implies that this function is convex partitionable w.r.t.~$G$, and hence the \DEGREE problem can be approximately solved (to within a constant factor) by using Algorithm \ref{alg:distributed_generalLP} with $\epsilon = 1/2$.
Next, we use the following rounding scheme proposed in \cite{CD16}: each edge $e \in E$ is included in the spanner with probability $x_e^{1/k}$. It is clear that this can be done in a constant number of rounds, and hence the overall algorithm takes $O(k \log n)$ rounds (by Theorem \ref{thm:distributed_LP}) in the $\mathcal{LOCAL}$ model. In \cite{CD16}, it was shown that this leads to a $\tilde{O}(\Delta^{(1-1/k)^2})$-approximation solution of the problem.
\end{proof}

\subsection{\DSNDISTANCE} \label{sec:DSN}

It is well-known that the centralized rounding of~\cite{Berman2011} for \DIRECTED is more general than is actually stated in their paper.  In particular, the randomized rounding for ``thin'' edges gives the same guarantee even when each demand has a possibly different distance constraint.  This fact was used, e.g., in~\cite{CDKL17} in their algorithms for \textsc{Distance Preserver}, \textsc{Pairwise $k$-Spanner}, and \DSNDISTANCE.  The difficulty in extending the algorithm of~\cite{Berman2011} is not in the LP rounding, but rather because the arborescence sampling technique used to handle thick edges in~\cite{Berman2011} (and in our Algorithm~\ref{alg:kspanner}) assumes that $n$ is a lower bound on the optimal cost.  This assumption is true for \DIRECTED, but false for variants where there might be a tiny number of demands.  However, it is easy to see that if we assume the demand graph is spanning (i.e., assume that every node is an endpoint of at least one demand) then the optimal solution must have at least $n/2$ edges, and hence we can again just use~\cite{Berman2011} to get a $\tilde O(\sqrt{n})$-approximation for \DSNDISTANCE as long as the demand graph is spanning.  

While this is in the centralized setting, since our algorithm for \DIRECTED is just a lightly modified distributed version of~\cite{Berman2011} (the only difficulty in the distributed setting is solving the LP, which is why that is the main technical contribution of this paper), we can easily modify it to give the same approximation for \DSNDISTANCE with spanning demand graphs.  The only change is that we use $D = \max_{(u,v) \in \mathcal S} L(u,v)$ instead of $k$ when solving the linear programming relaxation (using Theorem~\ref{thm:distributed_LP}) and when truncating the shortest-path arborescences that we sample (note that we have to assume that $D$ is global knowledge, which is reasonable for spanner problems and for \SLSN but may be less reasonable for other special cases of \DSNDISTANCE). This implies Theorem~\ref{thm:DSN-main}, and all of the interesting special cases (\SLSN, \textsc{Distance Preserver}, \textsc{Pairwise $k$-Spanner}, etc.) which it includes.

\iflipics
\else
\section{Conclusion}
In this paper we presented a distributed algorithm for solving distance-bounded network design convex programs in the \LOCAL model of distributed computation.  This is one of the few classes of convex programs (along with positive linear programs and a few other special cases) for which we now have distributed algorithms.  This class is particularly interesting since many state-of-the-art approximation algorithms for distance-bounded network design problems work by rounding one of these convex relaxations.  So if we can solve the relaxation in a distributed fashion then it is often straightforward to give a distributed approximation algorithm which simply solves the relaxation using our new distributed algorithm and then does the appropriate rounding.  Using this framework, we provide distributed approximation algorithms for a variety of problems (e.g., for \DIRECTED, \DEGREE, \BASIC[3], and \BASIC[4]) which use only polynomial-time computations at each node, achieve approximation ratios asymptotically equal to the centralized algorithm, and have very low round complexity.  Previous approaches to these problems either use exponential-time computations at local nodes, or if they use only polynomial-time computations require either significantly larger round complexity or give asymptotically worse approximations.  
\fi

\bibliography{DistributedSpanners}

\end{document}